\documentclass[11pt,a4paper]{article}
\usepackage[a4paper, margin=3cm]{geometry}
\usepackage[english]{babel}
\usepackage[utf8]{inputenc}
\usepackage[T1]{fontenc}
\usepackage{amsmath, amsfonts, amssymb, amsthm, mathtools}%
\usepackage{graphicx, xcolor}
\usepackage{hyperref}
\usepackage{parskip}
\usepackage{xfrac, dsfont}
\usepackage{siunitx}
\usepackage{lmodern}

\newtheoremstyle{style}
   {5pt}                   
   {5pt}                   
   {}                      
   {}                      
   {\normalfont\bfseries}  
   {.}                     
   {\newline}              
   {\textbf{\thmname{#1}\thmnumber{ #2}\thmnote{ (#3)}}} 

\theoremstyle{style}
\newtheorem{theorem}{Theorem}[section]
\newtheorem{definition}{Definition}[section]

\newtheorem{corollary}[theorem]{Corollary}
\newtheorem{proposition}[theorem]{Proposition}
\newtheorem*{remark}{Remark}

\newtheorem*{example}{Numerical example}

\renewcommand{\P}{\mathbb{P}}
\newcommand{\E}{\mathbb{E}}
\DeclareMathOperator{\GBi}{GBi}
\DeclareMathOperator{\Be}{Be}
\DeclareMathOperator{\Ga}{Ga}
\DeclareMathOperator{\Bi}{Bi}
\DeclareMathOperator{\supp}{supp}
\DeclareMathOperator{\Var}{Var}
\DeclareMathOperator{\Ber}{B}

\begin{document}

\title{A statistical noise model for a class of Physically Unclonable Functions} 
\author{
  Benjamin Hackl\\
  \small{Alpen-Adria-Universität Klagenfurt}\\
  \small{Department of Mathematics}\\
  \href{mailto:benjamin.hackl@aau.at}{\texttt{benjamin.hackl@aau.at}}
  \and
  Daniel Kurz\\
  \small{Alpen-Adria-Universität Klagenfurt}\\
  \small{Department of Statistics}\\
  \href{mailto:daniel.kurz@aau.at}{\texttt{daniel.kurz@aau.at}}
  \and
  Clemens Heuberger\\
  \small{Alpen-Adria-Universität Klagenfurt}\\
  \small{Department of Mathematics}\\
  \href{mailto:clemens.heuberger@aau.at}{\texttt{clemens.heuberger@aau.at}}
  \and
  Jürgen Pilz\\
  \small{Alpen-Adria-Universität Klagenfurt}\\
  \small{Department of Statistics}\\
  \href{mailto:juergen.pilz@aau.at}{\texttt{juergen.pilz@aau.at}}
  \and
  Martin Deutschmann\\ 
  \small{Technikon Forschungs- und}\\
  \small{Planungsgesellschaft mbH}\\
  \href{mailto:deutschmann@technikon.com}{\texttt{deutschmann@technikon.com}}
}
\date{}
\maketitle

\thispagestyle{empty}

\begin{abstract}
  The interest in ``Physically Unclonable Function''-devices has
  increased rapidly over the last few years, as they have several
  interesting properties for system security related applications
  like, for example, the management of cryptographic
  keys. Unfortunately, the output provided by these devices is noisy
  and needs to be corrected for these applications. 

  Related error correcting mechanisms are typically constructed on the
  basis of an equal error probability for each output bit. This
  assumption does not hold for Physically Unclonable
  Functions, where varying error probabilities can be observed. This
  results in a \emph{generalized binomial distribution} for the number
  of errors in the output. 

  The intention of this paper is to discuss a novel Bayesian
  statistical model for the noise of an especially wide-spread class
  of Physically Unclonable Functions, which properly handles the
  varying output stability and also reflects the different noise
  behaviors observed in a collection of such devices.
   Furthermore, we compare several different methods for
  estimating the model parameters and apply the proposed model to concrete
  measurements obtained within the CODES research project in order to
  evaluate typical correction and stabilization approaches.  
\end{abstract}

\section{Introduction and Preliminaries}

A simple but widely-used transmission model in coding theory is the
Binary Symmetric Channel (BSC). The general assumptions for this
transmission channel (cf.\ \cite{CodingTheory}) are as follows:  
\begin{itemize}
\item bitstrings (i.e.\ words consisting of the symbols $0$ and $1$) are
transmitted,
\item for each symbol the ``flip probability'' (i.e.\ the
probability that $0$ and $1$ is sent, but $1$ and $0$ is received,
respectively) is constant over all transmissions,
\item transmission errors occur independently.
\end{itemize}

Under these assumptions, the number of transmission
errors in an $n$-bit word can be modeled as a binomially distributed
random variable.

However, when investigating Physically Unclonable Functions (PUFs),
these assumptions do not apply. PUFs---and in particular
SRAM-PUFs---generally produce a noisy bitstring (response) for a given
input (challenge) where the flipping probability varies from bit to bit. Thus, a BSC 
is no longer an exact model for the noise of SRAM-PUFs (cf.\
\cite{helperData}). Another approach to model PUF error behavior
(concentrating on considerations regarding entropy) is pursued in
\cite{entropyAnalysis}.   

From a system security theoretic point of view, PUFs are very
interesting because they can be used to construct a challenge-response
mechanism without the need of having a ``master key''. Usually, the
master key is used to derive and verify responses to some given input
challenges. However, as PUFs basically are hardware challenge-response
mechanisms, a master key is not necessary. A general introduction to
the topic of PUFs can be found in \cite{PUFintroduction}.

In our setting, an SRAM-PUF consists of $n\in \mathbb{N}$ uninitialized
SRAM cells (cf.\ \cite{SRAMintroduction}). When
powering on the device, these cells either assume state $0$ or
$1$---and most of them do so in a very stable way. Based on a series
of measurements, we are able to identify a ``stable state'' for each
cell, which, in turn, is then used to compute the \emph{error
  probability} (or \emph{flipping probability}) for each cell. For
example, if a cell takes the state $1$ more often than the state $0$
in our measurements, we assess $1$ as the stable state of this
cell and $0$ is its error state. A query of all $n$ PUF bits is called
a \emph{PUF evaluation}. Otherwise, when investigating only a subset
of all PUF bits, then we speak of \emph{PUF responses}.

From a statistical point of view, the SRAM cells can be modeled by
means of independent Bernoulli distributed random variables $X_{j}
\sim \Ber(p_{j})$ for $j=1,2,\ldots, n$. That is, we have $\P(X_{j} =
1) = p_{j}$ and $\P(X_{j} = 0) = 1 - p_{j}$. The stability of an
SRAM-cell is then related to its error probability, where $p_{j} =
0.5$ means complete instability and $p_{j} = 0$ or $p_{j} = 1$ means
total stability. As we are only interested in the
stability of the cells, and not in the concrete values they assume,
our error probabilities can be bounded from above by $0.5$---this can
be enforced by choosing the stable state of the cell accordingly.

If these random variables are Bernoulli distributed with a
common error probability $p$, it is well known that the random number
of errors $E$ in a response of length $\ell$ follows a binomial
distribution with parameters $\ell$ and $p$, short $E\sim
\operatorname{Bi}(\ell, p)$. However, when investigating PUFs, the
situation is not quite as simple. Not all of the
SRAM cells are equally stable, which is indicated in
Figure~\ref{fig:errorAndWeight}, showing histograms for the error rate
and bit weight\footnote{The bit weight $w_{j}$ 
  is used to determine the stable states. It is computed
  separately for each bit as the proportion of ones among all
  measurements.} for a given set of SRAM-PUF measurements. One can see
that the majority of bits (about $80\%$) have an error rate of less
than $2\%$, while about $2.3\%$ of all SRAM cells show an error
probability greater than $40\%$, indicating strong instability of
these cells. Moreover, from the symmetry in the bit-weight histogram
one can see that the PUF actually is very balanced: there are
approximately the same number of bits assuming stable state $0$ as
there are assuming stable state $1$. 

\begin{figure}[ht]
  \centering
  \begin{minipage}[ht]{0.49\linewidth}
    \includegraphics[width=1\linewidth]{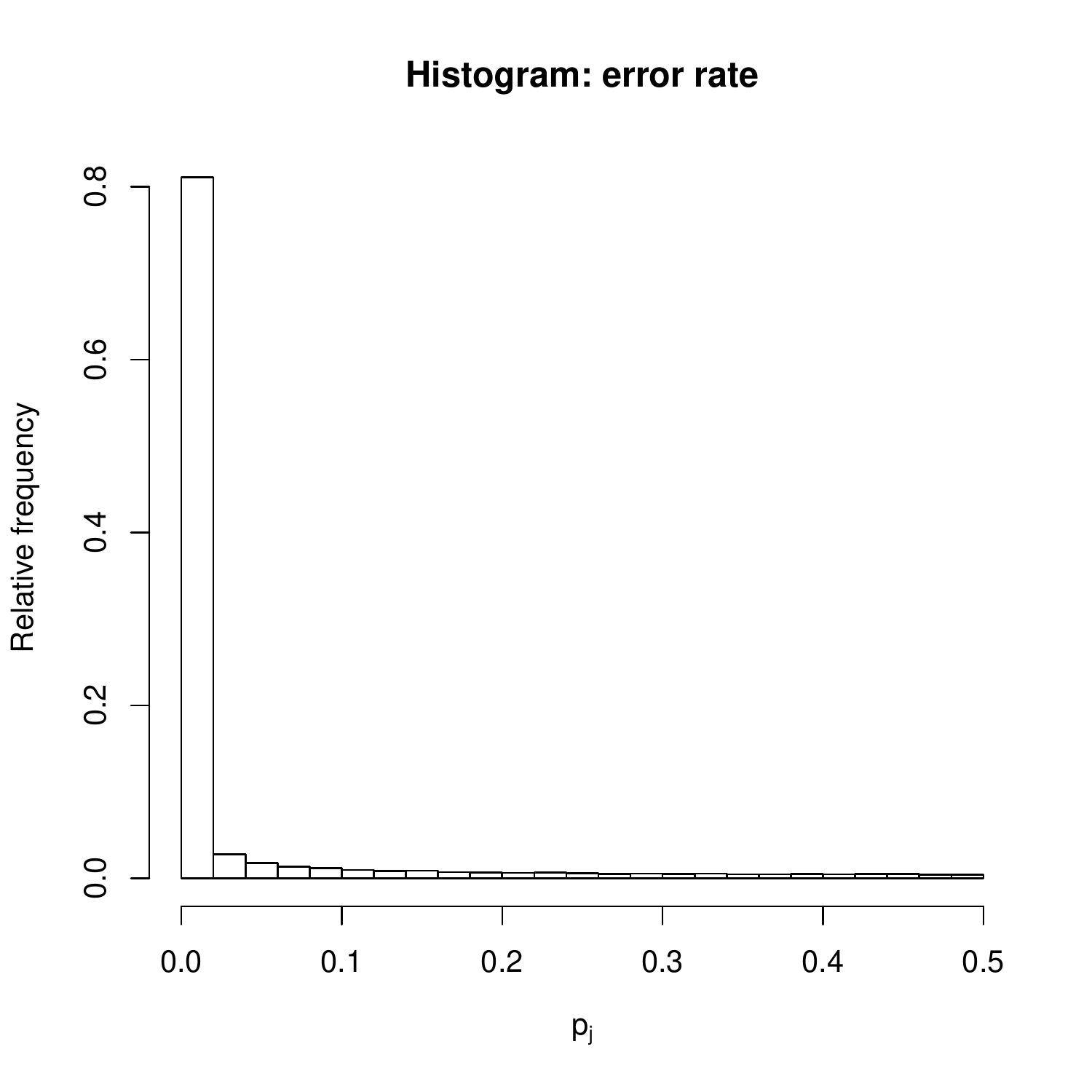}
  \end{minipage}\hfill
  \begin{minipage}[ht]{0.49\linewidth}
    \includegraphics[width=1\linewidth]{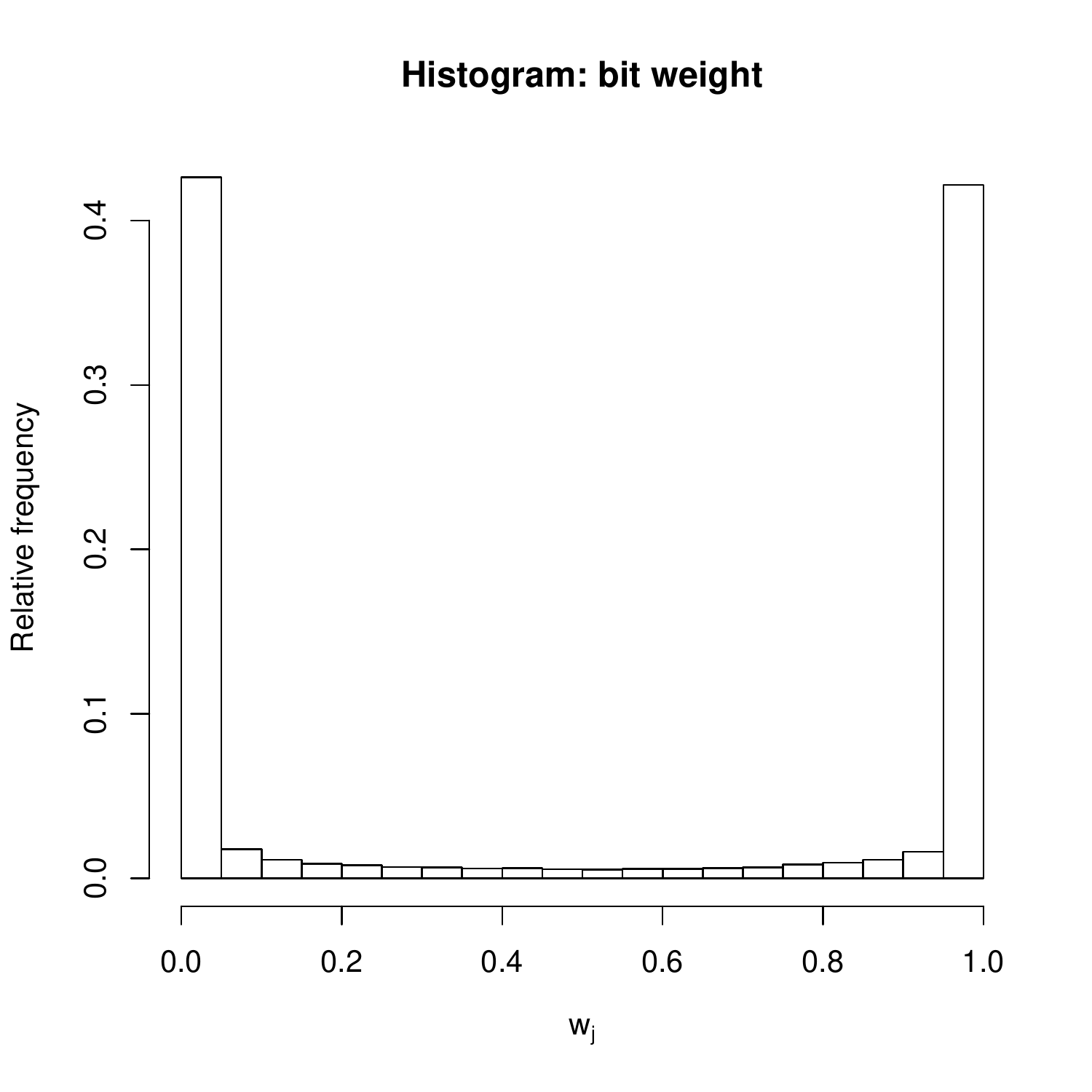}
  \end{minipage}
  \caption{SRAM cell stability: bit weights and error rates.}
  \label{fig:errorAndWeight}
\end{figure}

\begin{figure}[ht]
  \centering
  \includegraphics[width=0.9\linewidth]{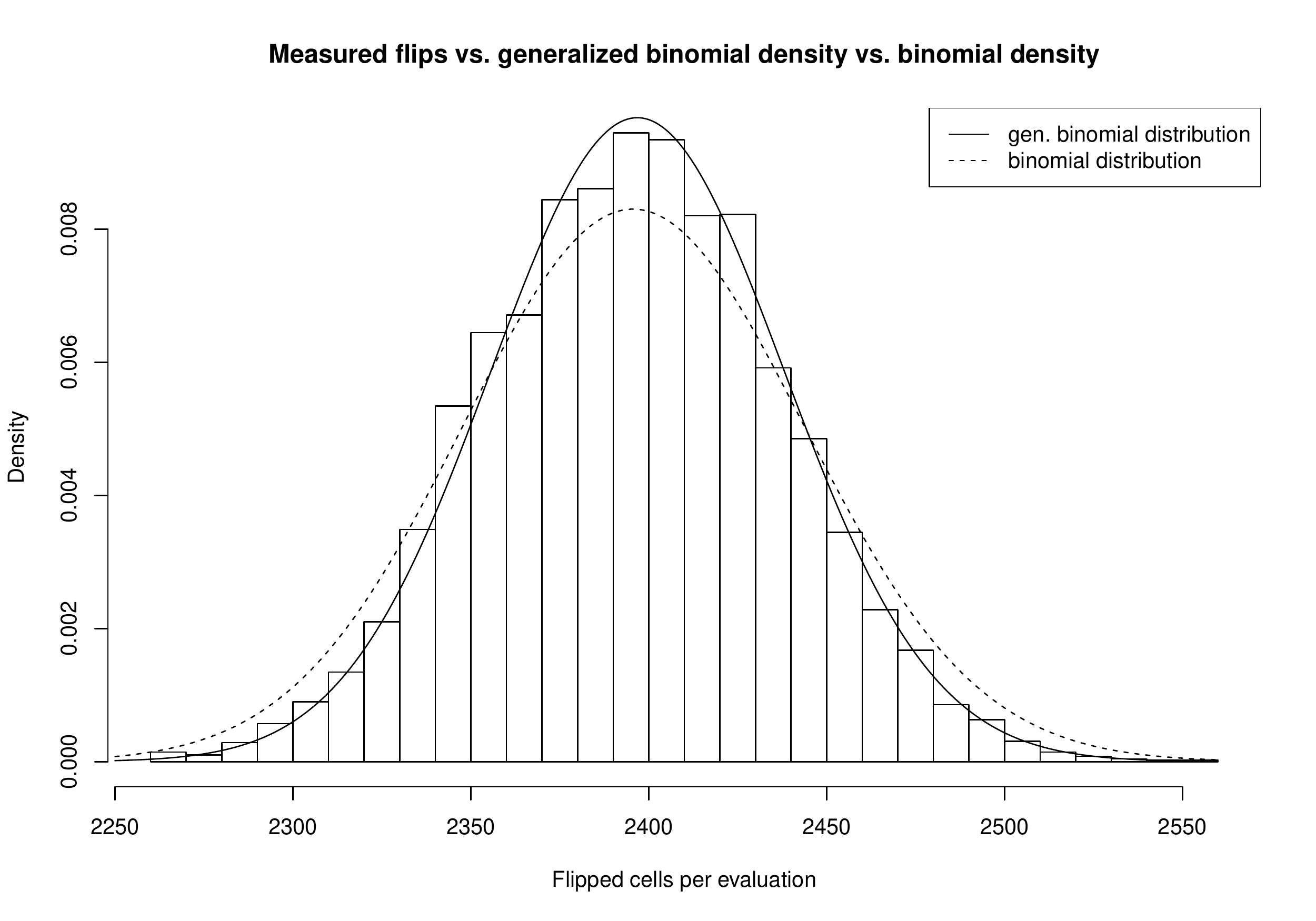}
  \caption{Goodness of fit: binomial vs.\ generalized binomial distribution.}
  \label{fig:densityComparison}
\end{figure}

In the case of varying error probabilities, the random number of
errors no longer follows a classical binomial distribution. 
Instead, we apply a \emph{generalized binomial
  distribution}, which yields a much better fit than the binomial
distribution for the number of flipped cells per PUF evaluation. This
is illustrated in Figure~\ref{fig:densityComparison}, showing a fitted
binomial and generalized binomial distribution for the histogram of
the number of flipped cells in a SRAM-PUF response of size $n = 2^{16}$.

In this paper, we investigate various properties of the
generalized binomial distribution as the proper model for the noise of
SRAM-PUFs (cf.\ Section~\ref{sec:gbin}). Furthermore, we develop a suitable
statistical model for evaluating the overall noise behavior of
SRAM-PUFs on the basis of measurements from several SRAM-devices (cf.\
Section~\ref{sec:model}). An empirical Bayesian approach is employed
to assess the parameters of the underlying generalized binomial
distribution. Several estimation techniques for determining the
hyperparameters in the empirical Bayesian model are compared within a
simulation study in Section~\ref{sec:applications}.

Finally, we extend the proposed model to concrete measurements obtained within
the CODES\footnote{\url{https://www.technikon.com/projects/codes}}
research project and use these concrete measurements to discuss
stabilization methods for SRAM-PUF responses on the basis of error
correction schemes and order statistics.

\section{Some properties of the generalized binomial distribution}\label{sec:gbin}

In this section, we will discuss some properties of the generalized
binomial distribution which arises naturally when investigating the
number of flipped SRAM cells per SRAM-PUF response.

In general, the binomial distribution originates in the context of
Bernoulli trials. In \cite{fisz}, the author calls the trials
under which the ``generalized binomial distribution'' originates, \emph{Poisson
  trials}\footnote{Therefore, the generalized binomial 
  distribution is also often called \emph{Poisson binomial
    distribution}.}. These are $n$ independent trials where the probability
that some event occurs in the $j$-th trial is $p_{j}$. Then, the
number of occurrences of the event within those $n$ trials follows a
generalized binomial distribution with parameter vector $\mathbf{p}
= (p_{1}, \ldots, p_{n})^{\top}$. Thus, we may define this distribution as
follows: 

\begin{definition}[Generalized binomial distribution]\label{def:GBi}
  Let $p_{1}$, $p_{2}$, \ldots, $p_{n}\in [0,1]$. Assume that the
  random variables $X_{j}\sim \Ber(p_{j})$ are independent for
  $j\in \{1,2,\ldots,n\}$. Then, the random variable $X :=
  \sum_{j=1}^{n} X_{j}$ follows a \emph{generalized binomial
    distribution} with parameter vector $\mathbf{p} =
  (p_{1},p_{2},\ldots, p_{n})^{\top}$. For short, we write $X\sim\GBi(p_{1}, p_{2},
  \ldots, p_{n})$.
\end{definition}

There is a number of elementary properties for the generalized
binomial distribution which follow immediately from the definition.

\begin{proposition}[Elementary properties]\label{prop:elementaryProperties}
  Let $X\sim \GBi(p_{1}, p_{2}, \ldots, p_{n})$. Then the following
  statements hold:
  \begin{enumerate}
  \item[(a)] The support of the random variable $X$ is contained in
    $\{0,1,\ldots, n\}$.
  \item[(b)] The probability mass function of $X$ is given by
    \[ \P(X = k) = \sum_{\substack{S\subseteq \{1,2,\ldots, n\} \\ |S| = k}}
    \prod_{s\in S} p_{s}~~\cdot~~ \prod_{\mathclap{s\in \{1,2,\ldots, n\}\setminus S}} (1-p_{s}).  \]
  \item[(c)] Expectation and variance of $X$ are given by
    \[ \E X = \sum_{j=1}^{n} p_{j} \quad \text{ and }\quad \Var(X) =
    \sum_{j=1}^{n} p_{j}\cdot (1-p_{j}),  \]
    respectively.
  \item[(d)] The characteristic function $\varphi_{X}(t)$ of $X$ is
    given by
    \[ \varphi_{X}(t) = \prod_{j=1}^{n} \left(1-p_{j} + p_{j}\cdot
      e^{it}\right).  \]
  \item[(e)] The generalized binomial distribution is a generalization
    of the binomial distribution: for $p_{1} = p_{2} = \cdots = p_{n}
    =:p$, the random variable $X$ follows a binomial distribution with
    parameters $n$ and $p$.
  \end{enumerate}
\end{proposition}
\begin{proof}
  With the notation of Definition~\ref{def:GBi}, (a) follows
  directly as we know $\supp (X_{j}) \subseteq \{0,1\}$ for all $j\in
  \{1,2,\ldots, n\}$. Therefore, the support of the sum is contained
  in $\{0,1,\ldots, n\}$.\\
  Based on the independence of $X_{j}$, the
  described probability mass function as stated in (b) follows
  immediately from
  \begin{align*} 
    \P(X = k) & = \sum_{\substack{S\subseteq \{1,2,\ldots, n\} \\ |S| = k}}
  \P\left(X_{s} = 1 \text{ for } s\in S \text{ and } X_{s} = 0 \text{
      for } s\in \{1,2,\ldots, n\}\setminus S\right) \\
  & = \sum_{\substack{S\subseteq \{1,2,\ldots, n\} \\ |S| = k}} \prod_{s\in S}
  p_{s} ~~\cdot~~ \prod_{\mathclap{s\in \{1,2,\ldots, n\}\setminus S}} (1-p_{s}).
  \end{align*}
  The remaining statements follow from the definition of $X$ as
  the sum of independently distributed random variables. We find
  \[ \E X = \E \left(\sum_{j=1}^{n} X_{j}\right) = \sum_{j=1}^{n} \E
  X_{j} = \sum_{j=1}^{n} p_{j},  \]
  by using the linearity of $\E$. We obtain $\Var(X_{j}) = p_{j}
  (1-p_{j})$ by using the linearity of the variance for independent
  random variables. Statement (d) is proved as
  the characteristic function of $X_{j}$ is $\varphi_{X_{j}}(t) =
  1-p_{j}+p_{j} e^{it}$ and because
  \[ \varphi_{X_{1}+X_{2}+\cdots +X_{n}}(t) = \prod_{j=1}^{n}
  \varphi_{X_{j}}(t)  \]
  holds for independent random variables $X_{1}$, \ldots,
  $X_{n}$. Finally, as the probabilities in (e) are all equal to $p$,
  $X$ is the sum of $n$ independent and identically distributed
  Bernoulli random variables---which is an alternative definition for
  the binomial distribution with parameters $n$ and $p$. This
  completes the proof.
\end{proof}

The \texttt{R} package ``\texttt{GenBinomApps}''
(cf.\ \cite{genBinomAppsPkg}) offers an
efficient implementation to compute the probability mass function
recursively. The theoretic background for this recursive computation
approach is covered in \cite{genBinomKurz}.

However, for very large dimensions of the probability vector, the
computation of the distribution function is quite expensive. In such
cases, approximation with a binomial distribution would be
desirable. In the following proposition we give some useful properties
of such an approximation.

\begin{proposition}[Binomial approximation]\label{prop:approx}
  Let $X\sim \GBi(p_{1},p_{2},\ldots, p_{n})$ and $Y\sim
  \operatorname{Bi}(n, p^{*})$. Then the following properties hold:
  \begin{enumerate}
  \item[(a)] The expectation of $X$ is equal to the expectation of $Y$
    if and only if $p^{*} = \overline{p} := \frac{1}{n} \sum_{j}
    p_{j}$.
  \item[(b)] If we have $p^{*} = \overline{p}$, then the inequality $\Var(X) \leq
    \Var(Y)$ holds for arbitrary parameters $p_{1}$, $p_{2}$, \ldots, $p_{n}
    \in [0,1]$ of $X$ and equality holds if and only if $p_{1} = p_{2}
    = \cdots = p_{n}$.
  \end{enumerate}
\end{proposition}
\begin{proof}
  The first statement follows immediately from
  Proposition~\ref{prop:elementaryProperties} and because of simple
  properties of the binomial distribution. We have $\E X =
  \sum_{j=1}^{n} p_{j}$ and $\E Y = n\cdot p^{*}$. Therefore, the
  relation
  \[ \E X = \E Y \quad \iff \quad p^{*} = \frac{1}{n} \sum_{j=1}^{n}
  p_{j}  \]
  follows immediately. The inequality from (b) can be proven by
  showing its equivalency to the Cauchy-Schwarz inequality. Let
  $\mathds{1}$ denote the $n$-dimensional vector of ones and
  $\mathbf{p} = (p_{1},p_{2},\ldots, p_{n})^{\top}$. Then we obtain
  \begin{align*}
    \Var(X) \leq \Var(Y) &\iff \sum_{j=1}^{n} p_{j}\cdot (1-p_{j})
    \leq n \cdot\overline{p}\cdot (1-\overline{p})\\ & \iff \sum_{j}
    p_{j}\cdot (1-p_{j}) \leq \bigg(\sum_{j} p_{j}\bigg)\cdot
    \bigg(1 - \frac{1}{n} \sum_{j} p_{j}\bigg)\\
    & \iff \sum_{j} p_{j} - \sum_{j} p_{j}^{2} \leq \sum_{j} p_{j} -
    \frac{1}{n} \bigg(\sum_{j} p_{j}\bigg)^{2} \\
    & \iff \bigg(\sum_{j} p_{j}\bigg)^{2} \leq n\cdot \sum_{j}
    p_{j}^{2}\\
    & \iff \left|\langle \mathds{1}, \mathbf{p}\rangle\right|^{2} \leq \langle
    \mathds{1},\mathds{1}\rangle \cdot \langle \mathbf{p}, \mathbf{p}\rangle,
  \end{align*}
  which is the Cauchy-Schwarz inequality for the
  Euclidean scalar product. Finally, equality holds in the
  Cauchy-Schwarz inequality if and only if the vectors $\mathds{1}$
  and $\mathbf{p}$ are linearly dependent, that is $\mathbf{p} =
  p\cdot \mathds{1}$. In this case, $p_{1} = p_{2} = \cdots = p_{n}  =
  p$ follows, which is equivalent to the fact that $X$ is binomially
  distributed with parameters $n$ and $p$.
\end{proof}

By property (b) of Proposition~\ref{prop:approx}, the variance of the
generalized binomial distribution equals the variance of the
approximating binomial distribution if and only if the distributions
coincide. In the following example, we investigate whether
the variance may be used to assess the goodness of fit for such an
approximation.

\begin{example}
  After performing some simulations (in which the input parameters of
  the generalized binomial distribution were generated from various beta
  distributions), we plotted the difference between the variances
  against the maximum error between the corresponding distribution
  functions. The result of one of these simulation batches ($N= 1000$
  generalized binomial distributions with $n=100$
  $\Be(1.5,1.8)$-distributed parameters each in the left plot, and
  following a $\Be(0.3,0.1)$-distribution in the right plot) is
  illustrated in Figure~\ref{fig:binomApprox}. The plots depict the
  relation between the difference of the variances and the maximum
  approximation error of the respective distribution functions
  mentioned above. 

  It is interesting to see that there is a very strong correlation
  (with correlation coefficient greater than $0.95$) between
  the difference of the variances and the maximal approximation error
  in our simulations when the $p_{j}$ concentrate around a single
  value. In this case, the variance difference also is significantly
  lower (as can also be seen in Figure~\ref{fig:binomApprox}).
\end{example}

\begin{figure}[ht]
  \centering
  \begin{minipage}[ht]{0.49\linewidth}
    \includegraphics[width=1\linewidth]{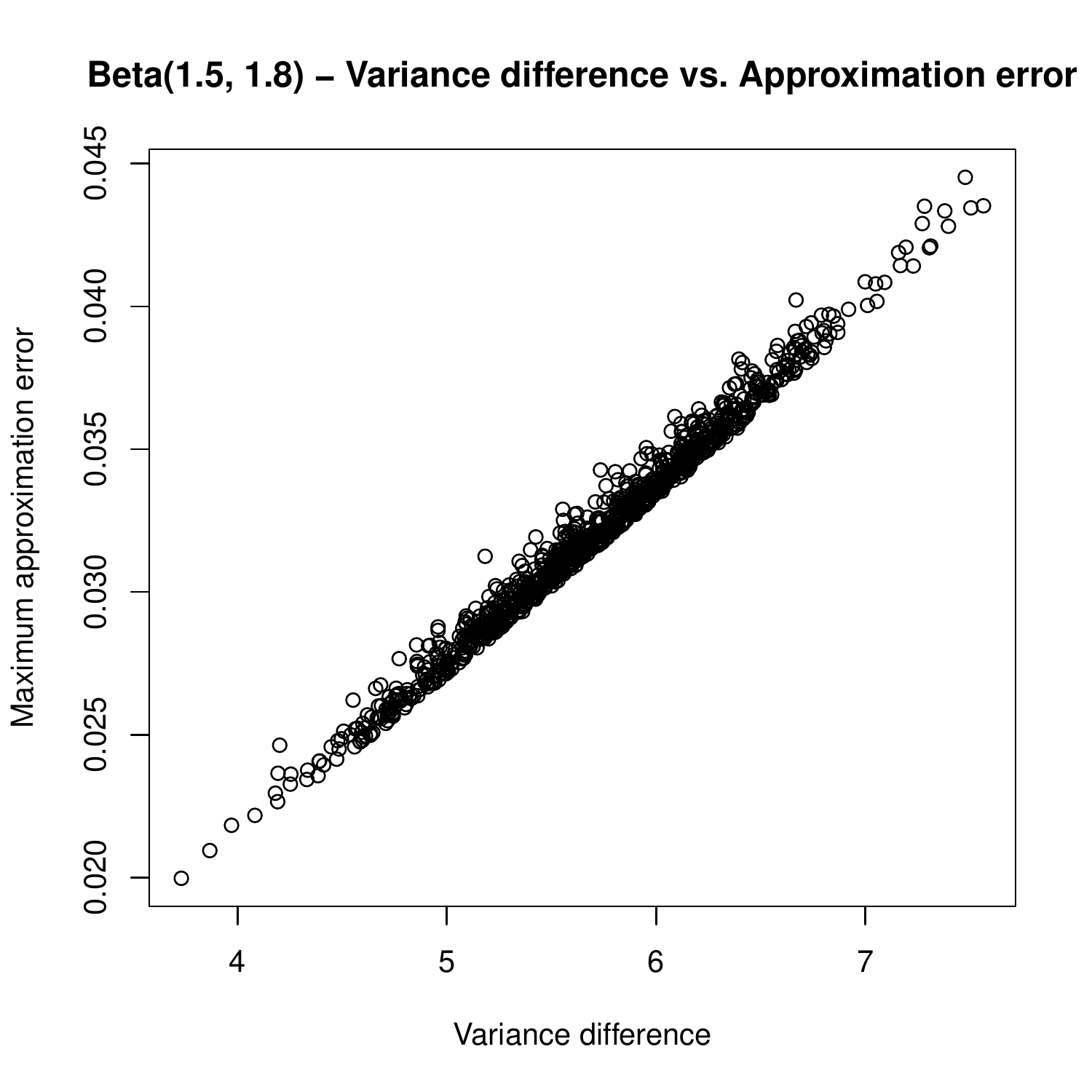}
  \end{minipage}
  \begin{minipage}[ht]{0.49\linewidth}
    \includegraphics[width=1\linewidth]{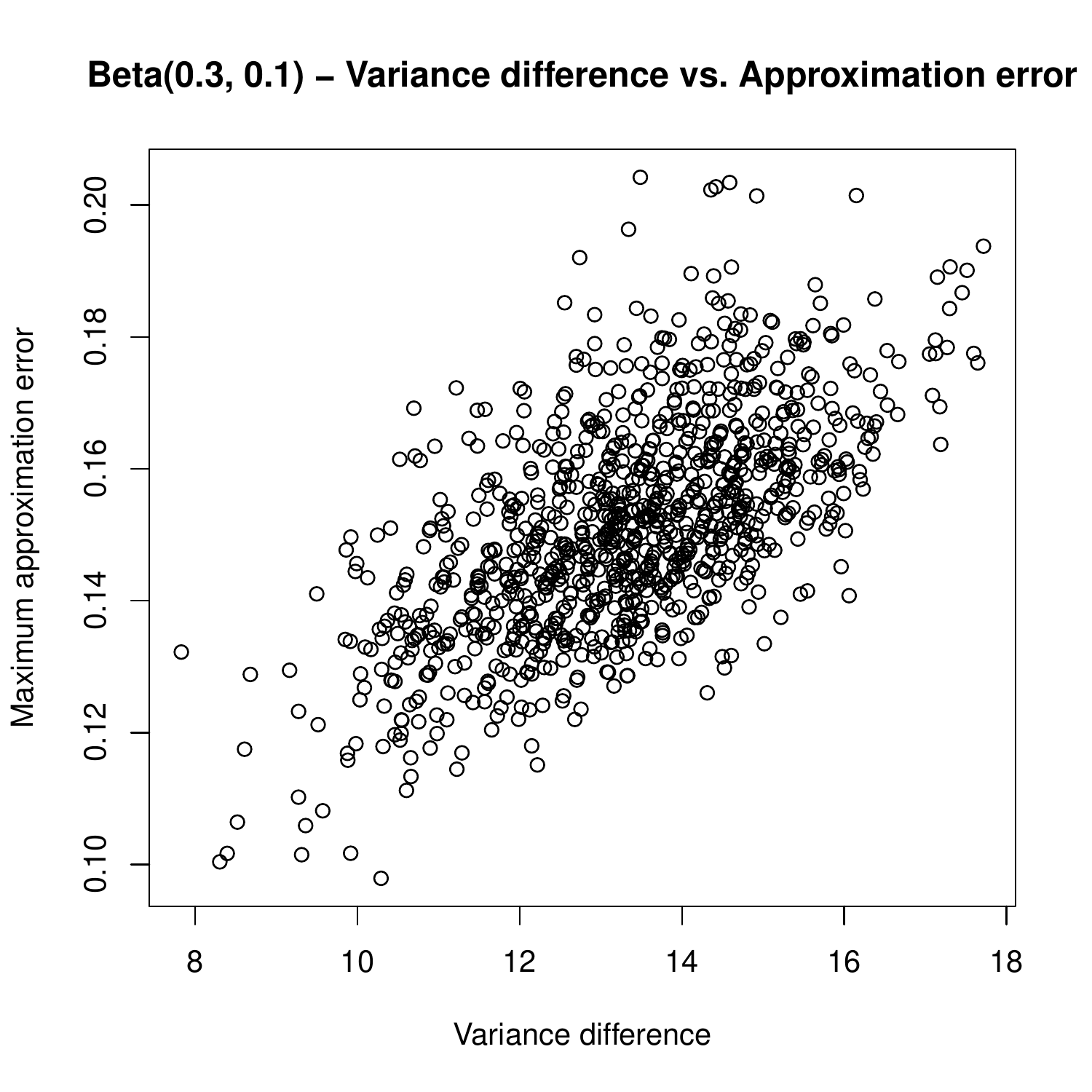}
  \end{minipage}
  \caption{Binomial approximation of the generalized binomial
    distribution: variance vs.\ maximum approximation error.}
  \label{fig:binomApprox}
\end{figure}

The following considerations require the notion of \emph{compound
  distributions}.

\begin{definition}[Compound distribution]\label{def:compound}
  Let $\mathbf{X}$ be a random vector with probability density
  function $p_{\mathbf{X}}(\mathbf{x}\,|\,\boldsymbol{\theta})$ depending on
  some random vector $\boldsymbol{\theta}$. Furthermore, let
  $G_{\boldsymbol{\theta}}(\boldsymbol{\vartheta})$ be the
  distribution function of $\boldsymbol{\theta}$. Then the
  probability density of the compound distribution of $\mathbf{X}$ with
  respect to $G$ is given by 
  \[ p_{\mathbf{X}}(\mathbf{x}) = \int_{\boldsymbol{\vartheta}}
  p_{\mathbf{X}}(\mathbf{x}\mid
  \boldsymbol{\vartheta})~dG_{\boldsymbol{\theta}}(\boldsymbol{\vartheta}).  \] 
\end{definition}

Now, let us assume that the error probabilities $p_{j}$ are realizations of
some random variable $\pi$ with distribution function $F_{\pi}$ and $\supp(
\pi) \subseteq [0,1]$. A quite interesting question is, how the compound
distribution of the generalized binomial distribution with respect 
to the parameters of the distribution of $\pi$ looks like. The following
theorem characterizes these compound distributions with respect to the
generalized binomial distribution.

\begin{theorem}[Compound distribution for the generalized binomial
  distribution]\label{thm:compound}
  Let $\pi_{1}$, $\pi_{2}$, \ldots, $\pi_{n}$ be independently and
  identically distributed random variables on $[0,1]$ with
  distribution function $F_{\pi|\boldsymbol{\vartheta}}$. Then the
  compound distribution $\GBi(\pi_{1},\pi_{2},\ldots, \pi_{n})$ under
  the parameter vector $\boldsymbol{\vartheta}$ is the binomial
  distribution $\Bi(n, \E(\pi\mid \boldsymbol{\vartheta}))$.
\end{theorem}
\begin{proof}
  It is easy to see that the expected value $\E(\pi\mid
  \boldsymbol{\vartheta})$ always exists and is contained in $[0,1]$,
  as the support of $\pi$ itself is contained in $[0,1]$. The expected
  value therefore is a valid second parameter for the binomial
  distribution.

  Now, let $X\mid \pi_{1},\ldots, \pi_{n}\sim \GBi(\pi_{1},\ldots,
  \pi_{n})$ and $\pi_{i}\stackrel{\text{iid}}{\sim} F_{\pi\mid
    \boldsymbol{\vartheta}}$. According to the definition of 
  compound distributions, the probability mass
  function $p_{X}(k\mid \boldsymbol{\vartheta})$ is then determined as
  follows:
  \begin{align*}
    p_{X}(k\mid \boldsymbol{\vartheta}) & = \sum_{\substack{S\subseteq
        \{1,\ldots, n\}\\ |S|=k}} \prod_{s\in S} 
    \left( \int_{p_{s}=0}^{1} p_{s}~dF_{\pi\mid \boldsymbol{\vartheta}}
    \right)\cdot \prod_{s\not\in S} \left( 1 - \int_{p_{s} = 0}^{1}
      p_{s}~dF_{\pi\mid \boldsymbol{\vartheta}} \right)\\
    & = \sum_{\substack{S\subseteq \{1,\ldots, n\}\\ |S|=k}} \prod_{s\in S}
    \E(\pi\mid \boldsymbol{\vartheta})\cdot \prod_{s\not\in S}
    (1-\E(\pi\mid \boldsymbol{\vartheta}))  = \binom{n}{k} \E(\pi\mid
    \boldsymbol{\vartheta})^{k}\cdot (1- \E(\pi\mid
    \boldsymbol{\vartheta}))^{n-k}.
  \end{align*}
  This is exactly the probability mass function of the $\Bi(n,
  \E(\pi\mid \boldsymbol{\vartheta}))$-distribution and therefore, proves
  the theorem.
\end{proof}

\begin{remark}
  \begin{enumerate}
  \item[(i)]  This result is in accordance with the result stated in property (b) of
    Proposition~\ref{prop:approx}: by forming the compound distribution,
    the variance also increases in general. 
  \item[(ii)] In a Bayesian context, the integrand in
    Definition~\ref{def:compound} can be identified with the
    likelihood (parameterized by $\boldsymbol{\vartheta}$), the
    integrating function $G_{\boldsymbol{\theta}}$ with the prior
    distribution of $\boldsymbol{\theta}$ and the left-hand side
    $p_{\mathbf{X}}(\mathbf{x})$ portrays the marginal distribution of
    $\mathbf{X}$, also known as prior predictive distribution. The
    increase in the variance of the compound distribution then
    accounts for the parameter uncertainty in a natural way. 
  \end{enumerate}
\end{remark}

\begin{corollary}[Properties of the compound distribution]
  Let $X\sim \GBi(\pi_{1},\pi_{2},\ldots, \pi_{n})$ with the 
  identically and independently distributed random variables
  $\pi_{1}$, $\pi_{2}$, \ldots, $\pi_{n}$ on $[0,1]$ with distribution
  function $F_{\pi|\boldsymbol{\vartheta}}$. Then the expected value
  and variance of $X$ are given by
  \[ \E X = n\cdot \E(\pi\mid \boldsymbol{\vartheta})\quad \text{ and
  }\quad \Var(X) = n\cdot \E(\pi\mid \boldsymbol{\vartheta})\cdot
  (1-\E(\pi\mid \boldsymbol{\vartheta})).  \]
  Furthermore, the characteristic function of $X$ has the form
  \[ \varphi_{X}(t) = \left[1-\E(\pi\mid \boldsymbol{\vartheta}) +
    \E(\pi\mid \boldsymbol{\vartheta})\cdot e^{it}\right]^{n}.  \]
\end{corollary}
\begin{proof}
  These statements follow immediately from Theorem~\ref{thm:compound}
  and some elementary properties of the binomial distribution.
\end{proof}

We will use these results in Section~\ref{sec:applications} in order to
estimate the probability that the SRAM-PUF noise cannot be
``corrected'' properly in a response of given length.

\section{Statistical Model for overall SRAM-PUF noise behavior}\label{sec:model}

In this section, we discuss a suitable
statistical model for the noise behavior of a set of different
SRAM-PUF devices, where the model parameters can be estimated from a
series of simple PUF evaluations.  The model is based on an
exploratory statistical analysis carried out within the CODES research
project. Furthermore, we propose several methods for assessing the
model parameters.

\subsection{Model development}\label{sec:modelDevel}

We propose a Bayesian model for the noise behavior of SRAM-PUFs:
assume that we have $m_{\operatorname{dev}}$ SRAM-PUF devices, where
device $i$ has $c_{i}$ SRAM-cells for $i=1,\ldots,
m_{\operatorname{dev}}$. Then, we model the noise of the $i$-th device
as a vector of $c_{i}$ independently distributed random variables
$X_{ij}\sim \Ber(p_{ij})$ such that 
\begin{align*}
  \P(\text{Cell } j \text{ in device } i \text{ flips}) & := \P(X_{ij}
= 1) = p_{ij}.
\end{align*}

Motivated by the results of an exploratory statistical analysis, we
further model these error probabilities to follow a scaled beta
distribution on the interval $[0,\sfrac{1}{2}]$. 

\begin{definition}[Scaled beta distribution]
  Let $a$ and $b$ be real numbers and $a < b$. If the random variable
  $P$ follows a beta distribution with parameters $\alpha$ and
  $\beta$, $P\sim \Be(\alpha, \beta)$, the random variable $Q = a +
  (b-a)\cdot P$ follows a scaled beta distribution on the interval
  $[a,b]$ with parameters $\alpha$ and $\beta$. For short, we write $Q \sim
  \Be_{[a,b]}(\alpha, \beta)$. 
\end{definition}

For the $i$-th device, we parametrize the beta distribution of the
parameters $(p_{ij})_{j=1}^{c_{i}}$ as
$\Be_{[0,\sfrac{1}{2}]}(2\delta_{i}\cdot K_{i}, (1-2\delta_{i})\cdot
K_{i})$, such that $\delta_{i}$ denotes the distribution's expected value and $K_{i}$ is a
shape parameter controlling the variance. In the next step, we assign
prior distributions to $\delta_{i}$ and $K_{i}$. More precisely, we
model $(\delta_{i})_{i=1}^{m_{\operatorname{dev}}}$ to
follow a scaled beta distribution (again scaled to the interval $[0,
\sfrac{1}{2}]$) with parameters $\alpha$ and $\beta$, and the shape
parameters $(K_{i})_{i=1}^{m_{\operatorname{dev}}}$ to follow a gamma
distribution with parameters $\kappa$ and $\lambda$. Altogether, we 
have
\begin{align*}
  X_{ij}\,|\, p_{ij} & \sim \Ber(p_{ij}) &  \text{for }i = 1,2,\ldots,
  m_{\operatorname{dev}}, &\quad j= 1,2,\ldots, c_{i},\\
  (p_{ij}\,|\,\delta_{i}, K_{i})_{j=1}^{c_{i}} & \sim \Be_{[0,\sfrac{1}{2}]}(2\delta_{i}\cdot K_{i}, (1-2\delta_{i})\cdot
  K_{i}) & \text{for } i=1,2,\ldots, m_{\operatorname{dev}},\\
  (\delta_{i}\,|\,\alpha,\beta)_{i=1}^{m_{\operatorname{dev}}} & \sim
  \Be_{[0,\sfrac{1}{2}]}(\alpha, \beta),\\
  (K_{i}\,|\,\kappa,\lambda)_{i=1}^{m_{\operatorname{dev}}} &\sim \Ga(\kappa, \lambda).
\end{align*}

\begin{remark}
  We choose this model over a simplified model without assumed
  distributions for the parameters $\delta_{i}$ and $K_{i}$ primarily
  because of the control we have over the mean error rate, as well as
  to reflect that different devices may have varying mean error rates.  
  Within this model, a variety of situations related to the SRAM-PUF
  production can be modeled and simulated.
\end{remark}

\subsection{Parameter estimation}\label{sec:paramEst}

Assuming we have $m_{\operatorname{dev}}$ SRAM-devices with $c_{i}$
SRAM-cells in the $i$-th device, the result of a series of $m_{i}$ measurements of
device $i$ is a vector $\mathbf{x}_{i}=(x_{i1},\ldots, x_{ic_{i}})^{\top}$,
where the component $x_{ij}$ denotes the number of measured error
states for the $j$-th cell of device $i$ and is $\Bi(m_{i},
p_{ij})$-distributed. Starting from these measurements, we wish to
estimate the parameters $\alpha$ and $\beta$ of the scaled beta
distribution $\Be_{[0,\sfrac{1}{2}]}(\alpha, \beta)$ modeling the mean
error rates $\delta_{i}$, and the parameters $\kappa$ and $\lambda$ of
the gamma distribution $\Ga(\kappa, \lambda)$ modeling the
distribution of the shape parameters $K_{i}$. Note that the component
$x_{ij}$ is a realization of a $\Bi(m_{i}, p_{ij})$ distribution.

\begin{remark}
  The joint posterior density function for our model is of the form
  \begin{multline*}
     f(\mathbf{p}, \boldsymbol{\delta}, \mathbf{K}, \alpha, \beta,
     \kappa, \lambda \mid \mathbf{x}) \propto \Bigg(\prod_{i=1}^{m_{\operatorname{dev}}}\Bigg(
  \prod_{j=1}^{c_{i}} f_{\Bi}(x_{ij}\mid p_{ij})\cdot
  f_{\Be}(p_{ij}\mid \delta_{i}, K_{i})\Bigg)\\
 \cdot
f_{\Be}(\delta_{i}\mid \alpha, \beta)\cdot f_{\Ga}(K_{i}\mid
\kappa,\lambda)\Bigg)\cdot f(\alpha,\beta,\kappa,\lambda),
  \end{multline*}
  where $f_{\Bi}(\,\cdot\,\mid p_{ij})$ denotes the density of the
  $\Bi(m_{i}, p_{ij})$ distribution, and
  $f_{\Be}(\,\cdot\,\mid \delta_{i},K_{i})$ and $f_{\Be}(\,\cdot\,\mid
  \alpha, \beta)$ denote the density of
  the $\Be_{[0,\sfrac{1}{2}]}(2\delta_{i}\cdot K_{i},
  (1-2\delta_{i})\cdot K_{i})$ and $\Be_{[0,\sfrac{1}{2}]}(\alpha,
  \beta)$ distribution, respectively. Moreover, $f_{\Ga}(\,\cdot\,|
  \kappa,\lambda)$ denotes the density of the
  $\Ga(\kappa,\lambda)$-distribution,
  $f(\alpha,\beta,\kappa,\lambda)$ denotes some joint prior of the
  parameters $\alpha$, $\beta$, $\kappa$ and $\lambda$, and
  $\mathbf{x}$ denotes the vector of all measurements. Due to the 
  (practically) very large number of parameters the simulation based
  on this posterior is computationally intractable. 

  To overcome this, we make use of an empirical Bayesian
  approach, meaning that we approximate the
  ``expensive'' posterior $f(\mathbf{p}, \boldsymbol{\delta},
  \mathbf{K}, \alpha, \beta, \kappa, \lambda \,|\, \mathbf{x})$ by the
  joint density $f(\mathbf{p}, \boldsymbol{\delta}, \mathbf{K}\mid
  \hat\alpha, \hat\beta, \hat\kappa, \hat\lambda, \mathbf{x})$ (which is also
  called ``pseudo posterior'') with empirically estimated parameters
  $\hat\alpha$, $\hat\beta$, $\hat\kappa$ and $\hat\lambda$. 
\end{remark}

We estimate the parameters according to the model hierarchy:
\begin{itemize}
\item from the measurements $x_{ij}$, we estimate the flipping
  probabilities $p_{ij}$,
\item from the estimated flipping probabilities, we estimate the
  parameters $\delta_{i}$ and $K_{i}$ for $i = 1,2,\ldots,
  m_{\operatorname{dev}}$,
\item and from these estimated parameters, we estimate the
  hyperparameters $\alpha$, $\beta$, $\kappa$ and $\lambda$.
\end{itemize}
The estimates of the hyperparameters are then depending on the
estimation techniques used for the different parameter
layers. For example, possible approaches are the \emph{method of moments},
\emph{maximum likelihood estimation} (MLE), or the construction of a
Bayes estimator.

By MLE for the $p_{ij}$, we obtain $\hat p_{ij} =
\frac{x_{ij}}{m_{i}}$. In this case, this MLE-estimator coincides with the 
estimator obtained by the method of moments. Another approach to estimate this parameter
(in the context of Bayesian statistics) is to choose the expected value of the
posterior obtained with respect to a scaled Jeffreys prior,
$p_{ij}\sim \Be_{[0,\sfrac{1}{2}]}(\sfrac{1}{2}, \sfrac{1}{2})$. The
posterior distribution is then given by
\[f(p_{ij}\,|\, x_{ij}) \propto \frac{p_{ij}^{x_{ij}}(1-p_{ij})^{m_{i}-x_{ij}}}{\sqrt{(2p_{ij})(1-2p_{ij})}}, \] forcing
us to compute the expectation numerically as
\[ \E(p_{ij}\mid x_{ij}) = \hat p_{ij} = \int_{0}^{\sfrac{1}{2}} p_{ij}\cdot
  f(p_{ij}\mid x_{ij})~dp_{ij},  \]
or to approximate it by Monte Carlo simulation. An advantage of this approach for
the estimation of the flipping probabilities is that it avoids an
underestimation of the flipping probabilities in the zero error
case. This is because the method of moments and MLE yield a flipping
probability of $0$ for cells without observed errors, which is not
realistic. Moreover, this approach allows us to compute sensible
credible intervals for  these probabilities, whereas the usual
confidence intervals based on the MLE would have zero lengths and thus
be meaningless. 

Given the estimates of $p_{i1}$, $\ldots$, $p_{i c_{i}}$, we can
estimate the parameters $\delta_{i}$ and $K_{i}$, either again by MLE,
by the method of moments or by a Bayesian approach similar to the one above,
where we use the joint noninformative prior 
\[p(\delta, K) \propto \frac{1}{K\cdot \sqrt{(2\delta)(1-2\delta)}}.\] 

After estimating the $\delta_{i}$ and $K_{i}$, we may use these values
to obtain estimations for the hyperparameters $\alpha$, $\beta$,
$\kappa$ and $\lambda$. As we want to avoid proposing more priors for
these parameters, we will use either MLE or the method of moments.

\section{Results and Applications}\label{sec:applications}

On the basis of the statistical model proposed in the previous
section, we will determine a suitable method for parameter estimation
of this model in Section~\ref{sec:simStudy} by comparing possible
approaches in a simulation study. Afterwards, in
Section~\ref{sec:appData}, we will estimate the parameters of our
model (based on the superior estimation method from the simulation
study) for real measurements from the CODES project. Finally, in Section~\ref{sec:corrAndRed}, we will use
the posterior predictive distribution based on our real data to
evaluate approaches to correct and stabilize the SRAM-PUF responses. 

\subsection{Simulation study}\label{sec:simStudy}

We are interested in comparing different parameter estimation methods
as discussed in Section~\ref{sec:paramEst} for the proposed
statistical model. In order to choose the ``best'' estimation method,
we will estimate these hyperparameters from simulated data with known
hyperparameters. The quality of these estimation methods will then be
compared by the value of a quadratic loss function for the parameter
vector: $L(\theta, \hat\theta) = \|\theta - \hat\theta\|^{2}$, where
$\theta\in \{(\alpha, \beta)^{\top}, (\kappa,\lambda)^{\top}\}$. Note
that we are especially interested in a good estimation of the
parameters $\alpha$ and $\beta$ of the beta distribution modeling the
mean error rates.

We will generate the simulation data ($m_{\operatorname{dev}}=20$
SRAM-devices with $c = 10000$ cells each and $m = 500$ simulations per
device) from the following parameters\footnote{These parameters are
  roughly based on parameters we used for testing in the CODES project.}: 
\[ \alpha = 100,\quad \beta = 900,\quad \kappa = 800,\quad \lambda =
900.  \]

Concretely, there are $8$ methods of parameter estimation we will
compare. These methods originate from the different possibilities to
estimate the various parameter hierarchies. Let $x_{ij}$ denote the
number of assumed error states of the $j$-th cell in device $i$.
\begin{itemize}
\item The flipping probabilities $p_{ij}$ can be estimated by the
  method of moments (which, in this case, coincides with maximum
  likelihood estimation) by $\hat p_{ij} = \frac{x_{ij}}{m}$, or by
  computing the Bayes-estimator with respect to the Jeffreys
  $\Be_{[0, \sfrac{1}{2}]}(\sfrac{1}{2}, \sfrac{1}{2})$-prior. 
\item The parameters $\delta_{i}$ and $K_{i}$ can be estimated either
  by the method of moments, yielding the estimators
  \[ \hat \delta_{i} =
  \frac{1}{c} \sum_{j=1}^{c} \hat p_{ij},\quad \hat K_{i} =
  \frac{2\hat\delta_{i} (1-2\hat\delta_{i})}{\frac{4}{c-1}\cdot
    \sum_{j=1}^{c} (\hat p_{ij} - \hat\delta_{i})^{2}} - 1,  \]
  by maximum likelihood estimation with the \texttt{R}-package
  \texttt{maxLik} (cf.\ \cite{maxLikPkg}), or by using a Bayesian
  approach and computing the mode of the joint (independence)
  posterior distribution subject to the Jeffreys $\Be_{[0, \sfrac{1}{2}]}(\sfrac{1}{2},
  \sfrac{1}{2})$-prior for $\delta_{i}$ and the non-informative
  $\frac{1}{\vartheta}$-prior for $K_{i}$.

\item Finally, the hyperparameters $\alpha$, $\beta$ from the proposed
  beta distribution of the $\delta_{i}$ and the parameters $\kappa$, $\lambda$ from
  the proposed gamma distribution of the $K_{i}$ can be estimated by
  the method of moments, which yields
  \[ \hat\alpha = 2\overline{\delta}\left(\frac{2\overline{\delta}
      (1-2\overline{\delta})}{4v_{\delta}} - 1\right),\quad 
  \hat\beta = (1-2\overline{\delta})\left(\frac{2\overline{\delta}
      (1-2\overline{\delta})}{4v_{\delta}} - 1\right),\quad
  \hat\kappa = \frac{\overline{K}^{2}}{v_{K}},\quad \hat\lambda =
  \frac{\overline{K}}{v_{K}}, \]
  where $\overline{\delta}$, $v_{\delta}$, $\overline{K}$ and $v_{K}$
  denote the means and sample variances of the $\hat\delta_{i}$ and
  the $\hat K_{i}$, respectively---or by maximum likelihood
  estimation. As we want to refrain from proposing another set of
  priors for these parameters, we will not use Bayesian estimation for
  $\alpha$, $\beta$, $\kappa$ and $\lambda$. 
\end{itemize}

In order to compare the various possible combinations of parameter
estimation methods, we used a quadratic loss function to measure the distance
from the original parameters. After performing $10000$ simulations,
and investigating the respective mean losses (which can be found in
Table~\ref{tab:simStudy}), we find that estimating the $p_{ij}$ and
the parameters $\delta_{i}$ and $K_{i}$ with Bayesian methods, as well
as the hyperparameters $\alpha$, $\beta$, $\kappa$ and $\lambda$ with
maximum likelihood estimation yields the lowest overall loss (where
the estimates of all four hyperparameters are taken into account) as
well as the lowest loss for just the parameters $\alpha$ and $\beta$
of the beta distribution. However, the lowest loss for the parameters
$\kappa$ and $\lambda$ of the gamma distribution originates from
estimating the $p_{ij}$ with Bayesian methods, but using the method of
moments to estimate everything else.  

Note that even although the approach where we estimate $p_{ij}$ as
well as $\delta_{i}$ and $K_{i}$ by Bayesian means, and the remaining
parameters by MLE yields the lowest loss function with respect to the
parameters $\alpha$ and $\beta$, this estimator is rather conservative
with respect to the expected mean error rate $\E\delta =
\frac{1}{2}\cdot \frac{\alpha}{\alpha + \beta} = 0.05$: taking the
average value over the $\frac{1}{2} \cdot \frac{\hat\alpha}{\hat\alpha
+ \hat\beta}$ obtained in the simulation study (where $\hat\alpha$ and
$\hat\beta$ have been constructed as mentioned above) yields a value
of $0.0618$, which can be contributed to the high number of
observed very unstable bits. In order to
cover the occurrence of such bits, the expected average error rate is
increased in the MLE-estimation. Therefore, estimating
the parameters with this approach yields a model, which possesses
a certain ``robustness'' regarding a decline of the PUF's
stability. The most accurate approximation of the expected mean error
rate is obtained by estimating all parameters by the method of moments.

\begin{table}[ht]
  \centering
  \begin{tabular}[ht]{l|l|l|r|r}
    $\hat p_{ij}$ & $\hat\delta_{i}$, $\hat K_{i}$ & $\hat\alpha$,
    $\hat\beta$, $\hat\kappa$, $\hat\lambda$ & Mean loss ($\alpha$, $\beta$) & Mean
    loss ($\kappa$, $\lambda$) \\ \hline 
    Moments/MLE & Moments & Moments & $10684999.5$ & $202115.0$ \\
    Moments/MLE & MLE & MLE & $6573128.9$ & $397729.0$ \\
    Bayes & Moments & Moments & $11077285.5$ & $\underline{200906.9}$ \\
    Bayes & Moments & MLE & $22138561.9$ & $236938.5$ \\
    Bayes & MLE & Moments & $15936758.9$ & $917774.9$ \\
    Bayes & MLE & MLE & $7831878.6$ & $750779.0$ \\
    Bayes & Bayes & Moments & $14402261.7$ & $968651.6$ \\
    \textbf{Bayes} & \textbf{Bayes} & \textbf{MLE} & $\underline{3042068.4}$ & $1895868.2$ 
  \end{tabular}
  \caption{Average loss function values from the simulation study
    ($10000$ simulations).}
  \label{tab:simStudy}
\end{table}

\subsection{Parameter estimation for real data}\label{sec:appData}

We are investigating measurements originating from
$m_{\operatorname{dev}} = 15$ different SRAM-PUF devices, each of them
with $c = 2^{16}$ SRAM cells. Note that our given
measurements were carried out on ASICs that have been manufactured in
TSMC $65 \si{nm}$ CMOS technology within a European multi-project
wafer run. The ASIC has been designed within the
UNIQUE\footnote{\url{http://www.unique-project.eu}} research project.

For each device, we have $340$
evaluations---however, the first $50$ measurements are discarded
because they were conducted during an aging process, which slightly
changed the behavior of the SRAM-PUFs. Afterwards, during the
remaining $m = 290$ measurements, the devices are stable again,
meaning that we will focus our analysis on these measurements.

For the parameter estimation, we will follow the results of the
simulation study, meaning that we will estimate the flipping
probabilities $p_{ij}$ and the parameters $\delta_{i}$ and $K_{i}$ by
the Bayesian approaches described above, and the four parameters
$\alpha$, $\beta$, $\kappa$ and $\lambda$ by maximum likelihood estimation.

This results in the following estimates:
\[ \hat\alpha = 9378.324\quad \hat\beta = 81409.79,\quad \hat\kappa = 7166.669,\quad
\hat\lambda = 3965.296. \]

In Figure~\ref{fig:paramDens}, we plotted histograms for the
(respectively) estimated $\delta_{i}$ and $K_{i}$, as well as the
densities of the proposed probability distributions.

\begin{figure}[ht]
  \centering
  \begin{minipage}[ht]{0.49\linewidth}
    \includegraphics[width=1\linewidth]{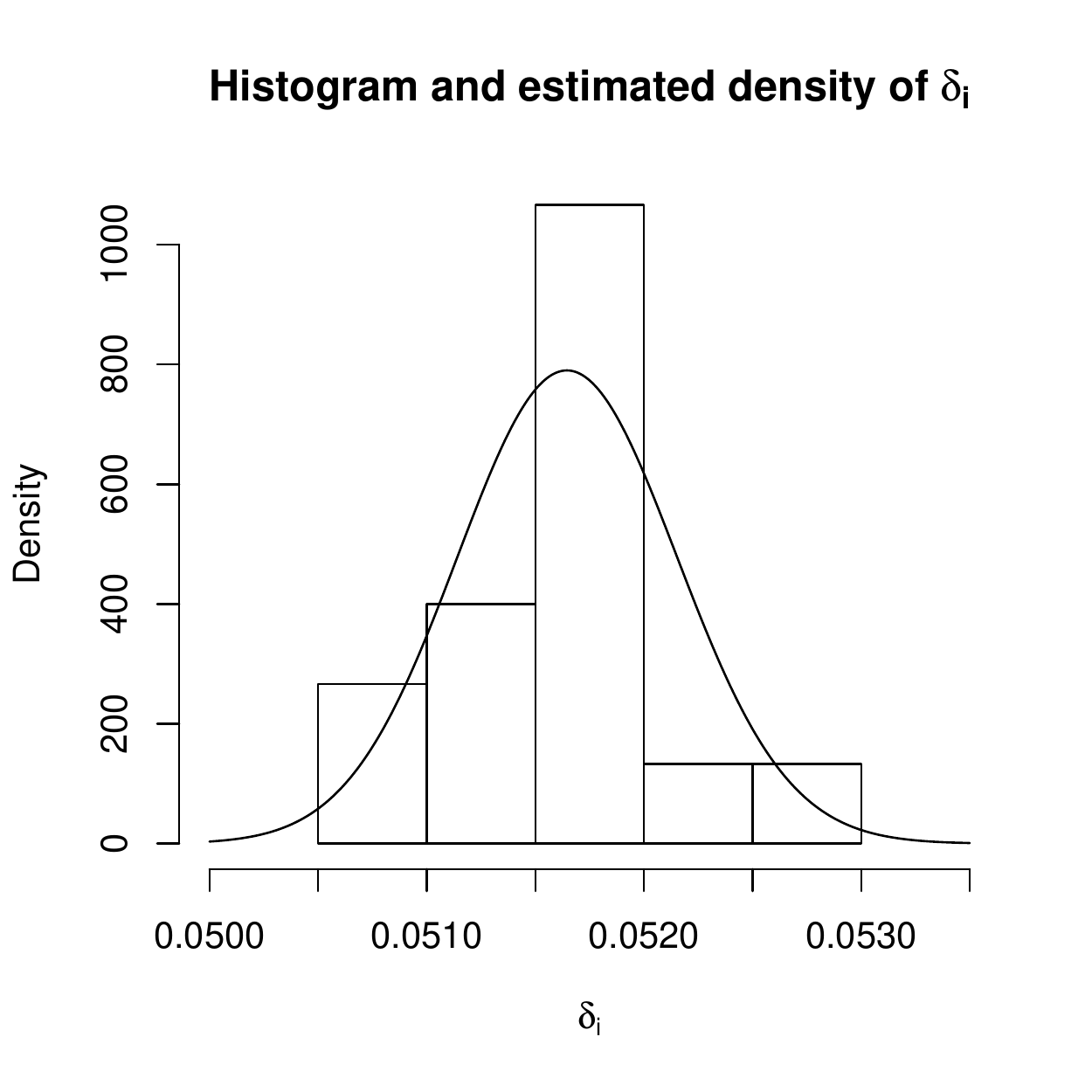}
  \end{minipage}\hfill
  \begin{minipage}[ht]{0.49\linewidth}
    \includegraphics[width=1\linewidth]{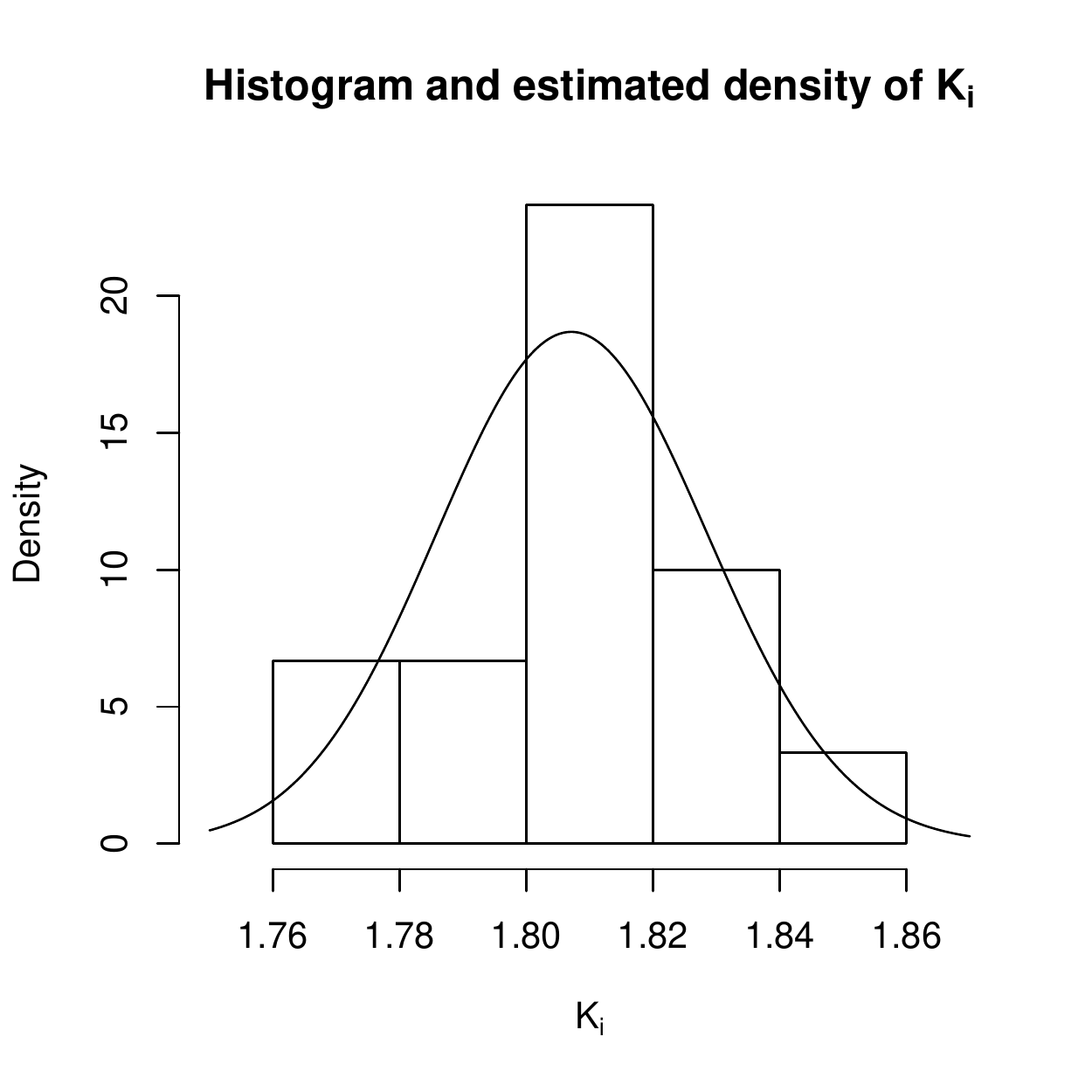}
  \end{minipage}
  \caption{Histograms and estimated densities for $\delta_i$ and $K_i$.}
  \label{fig:paramDens}
\end{figure}

Furthermore, we may compute the expected values for the parameters
$\delta$ and $K$, and construct credible intervals. From the estimated
parameters we obtain
\[ \E \delta = \frac{1}{2} \frac{\alpha}{\alpha + \beta}
\approx 0.05165,\quad \E K = \frac{\kappa}{\lambda} \approx 1.8073,  \]
and empirical $95\%$ credible intervals
\[ \delta\in [0.05066, 0.05264],\quad K\in [1.76574, 1.84943].  \]

\subsection{Error correction and reduction}\label{sec:corrAndRed}

In practice, we are interested in stabilizing the responses of a PUF
such that it can be used for system security related aspects like
constructing a challenge-response system without the need of storing a
master key. We want to present two general approaches (concentrating
on error correction and error reduction) for stabilizing these
responses and use our proposed model in order to evaluate their
effectiveness in specific examples.

Based on the parameters $\alpha$, $\beta$, $\kappa$ and $\lambda$
estimated in the previous section, we may investigate the
posterior-predictive distribution for the flipping probabilities. Its
density function is given by
\begin{equation*}
f(p\mid \hat\alpha,\hat\beta,\hat\kappa,\hat\lambda, \mathbf{x}) = \int_{\delta} \int_{K}
f_{\Be}(p\mid \delta, K)\cdot f_{\Be}(\delta\mid \hat\alpha,\hat\beta)\cdot f_{\Ga}(K\mid \hat\kappa,\hat\lambda)~dK~d\delta
\end{equation*}
for $p\in (0,\sfrac{1}{2})$ and $0$ otherwise. By simulation of a
sample of size $100000$, we obtain $\overline{p} = 0.05187$ as an
approximation for the expected value of the posterior-predictive
distribution under the empirically estimated parameters $\hat\alpha$,
$\hat\beta$, $\hat\kappa$ and $\hat\lambda$ from above. 

We are interested in the number of errors in an $\ell$-bit SRAM-PUF
response, where the bit-wise error probabilities are distributed
according to the posterior-predictive distribution from above. For
fixed error probabilities, the quantity of errors follows a
generalized binomial distribution. Due to Theorem~\ref{thm:compound},
the resulting compound distribution is a $\Bi(\ell,
\overline{p})$-distribution. This distribution can now be used to
compute the probability that a given error correction mechanism fails.

\begin{example}
  Assume that an SRAM-PUF is embedded within a construction which
  allows the correction of up to $239$ bits in responses of length
  $\ell = 1953$. Following the model above, the expected number of errors is
  $101.3101$, and the probability that the PUF does not work properly
  (i.e.\ the probability that more than $239$ errors occur) is
  negligibly small (less than $10^{-20}$).
\end{example}

Instead of designing powerful mechanisms for error correction which are able to
compensate for the noise an SRAM-PUF produces, another approach is to
``ignore'' SRAM-cells for which a high flipping probability is 
known or has been estimated. Assuming that we are investigating $\ell$
bit SRAM-PUF responses, it is an interesting question how the removal
of $r\ll \ell$ unstable bits influences the noise behavior.

In general, a good measure to judge the effect of ignoring the $r$
``worst'' bits is the average flipping probability of the
remaining cells, that is if $p_{1}$, \ldots, $p_{\ell}$ are the
respective flipping probabilities, and $p_{(1)}\leq \cdots \leq
p_{(\ell)}$ denote the related ordered probabilities, then we are
interested in $\frac{1}{\ell - r}\cdot \sum_{j=1}^{\ell - r}
p_{(j)}$. The assumption that these probabilities are realizations of
identical and independently distributed random variables leads us to
\emph{order statistics}. 

\begin{definition}[Order statistics]
  Let $X_{1}$, \ldots, $X_{n}$ be identical and independently
  distributed random variables with respect to some distribution
  $X$. Then the \emph{ordered} random variables $X_{(k)}$ with
  $X_{(1)} \leq X_{(2)} \leq \cdots \leq X_{(n)}$ are called
  $k$-th smallest \emph{order statistic} of size $n$ with respect to $X$.
\end{definition}

The following theorem states a central result from the theory of order
statistics, a proof can be found in \cite{orderStatLit}.

\begin{theorem}[Order statistics and the beta distribution]
  The density function of the $k$-th smallest order statistic
  $U_{(k)}$ of size $n$ with respect to the uniform distribution on
  the interval $[0,1]$ is
  \[ f_{U_{(k)}}(u) = \begin{cases} \frac{n!}{(k-1)!\,
        (n-k)!}\cdot u^{k-1}\cdot (1-u)^{n-k} & \text{for } u\in
      (0,1),\\ 0 & \text{else,} \end{cases}  \]
  which is the density of a beta distribution with parameters $k$ and
  $n-k+1$. Therefore, we have $U_{(k)}\sim \Be(k, n-k+1)$.
\end{theorem}

By the technique of \emph{Probability Integral Transform}, this result
may be used to express the density function of an arbitrary
continuous random variable $X$ with distribution function $F_{X}$ and
density function $f_{X}$: in this case, we obtain
\[ f_{X_{(k)}}(x) = \frac{n!}{(k-1)!\, (n-k)!}\cdot 
\left[F_{X}(x)\right]^{k-1}\cdot \left[1 - F_{X}(x)\right]^{n-k}\cdot
f_{X}(x)   \]
for the density function of the $k$-th smallest order statistic of
size $n$ with respect to the distribution of $X$.

\begin{remark}
  Note that as for $1\leq j < k \leq n$ the relation $X_{(j)}\leq
  X_{(k)}$ holds, the order statistics are \emph{not} independently
  distributed any more. This means that if we would like
  to compute some compound distribution of, for example, the
  generalized binomial distribution and these order statistics with
  respect to a (scaled) beta distribution (which could be used
  to predict the probability that an $\ell$-bit PUF response with the
  $r$ most unstable bits removed is still too noisy for correction),
  we would have to consider the respective joint densities (which can
  be found in \cite{orderStatLit}) when integrating over the respective
  parameters. As we focus on the scaled beta distribution, the arising
  integrals cannot be computed analytically (mainly because of the
  occurring products of incomplete beta functions)---however, by
  simulating the procedure, i.e.\ generating $N$ beta-distributed
  samples of size $\ell$ and removing the $r$ largest values, we may
  still give good approximations for some interesting parameters
  within specific examples.
\end{remark}

In the setting above, we are interested in the expected flipping
probability, averaged over all remaining SRAM-cells, that is
\[ \E\left[\frac{1}{\ell - r}\cdot \sum_{j=1}^{\ell - r}
  \pi_{(j)}\right] = \frac{1}{\ell -r}\cdot \sum_{j=1}^{\ell - r} \E
\pi_{(j)}.  \]
Thus, we primarily want to compute the expected value of the order
statistics. For the same reasons as mentioned in the remark above,
this is not possible analytically for arbitrary
distributions. However, for a special case of the scaled beta
distribution it is actually possible, and that is for $\Be(\alpha,
\beta)$ with $\alpha = 1$ or $\beta = 1$.

\begin{proposition}\label{prop:simplifiedOrderStat}
  Let $\pi_{1}$, \ldots, $\pi_{n}$ be independent and identically
  distributed random variables following a $\Be_{[a,b]}(\alpha,
  1)$-distribution. The expected value of the $k$-th smallest order
  statistic of size $n$ is then given by
  \[ \E \pi_{(k)} = a + (b-a)\cdot \frac{B(k + \sfrac{1}{\alpha}, n - k
    + 1)}{B(k, n-k+1)}.  \]
\end{proposition}
\begin{proof}
  As scaled beta distributed random variables are affine-linearly
  transformed beta distributed random variables (which also holds for
  the related order statistics), and as the expectation is a linear
  operator, we may concentrate on the case $\pi_{1}$, \ldots,
  $\pi_{n}\stackrel{\text{iid}}{\sim} \Be(\alpha, 1)$.
  
  The density function of these random variables is given by $f(x) =
  \mathds{1}_{(0,1)}(x)\cdot\alpha \cdot x^{\alpha-1}$, and thus the
  distribution function has the shape
  \[ F(x) = \begin{cases} 0 & \text{for } x\leq 0, \\
    x^{\alpha} & \text{for } 0 < x < 1, \\ 
    1 & \text{for } x\geq 1.
  \end{cases}  \]
  As mentioned above, by \emph{Probability Integral Transform}, the
  density of the $k$-th smallest order statistic of size $n$ with
  respect to the $\Be(\alpha, 1)$-distribution has the form
  \begin{align*} 
    f_{\pi_{(k)}}(x) & = k \binom{n}{k}\cdot [F(x)]^{k-1}\cdot [1- F(x)]^{n-k}\cdot f(x)\\
    & = \mathds{1}_{(0,1)}(x)\cdot \alpha k \binom{n}{k} \cdot
    x^{\alpha k - 1}\cdot (1 - x^{\alpha})^{n-k}.
  \end{align*}
  The expected value thus reads
  \begin{equation*}
    \E \pi_{(k)}  = \alpha k \binom{n}{k}\cdot \int_{0}^{1} x^{\alpha
    k} (1-x^{\alpha})^{n-k}~dx,
  \end{equation*}
  which, after a change of variables $t = x^{\alpha}$, becomes
  \begin{equation*}
  \E \pi_{(k)}  = \frac{\int_{0}^{1} t^{k + \sfrac{1}{\alpha} - 1}
    (1-t)^{n-k}~dt}{B(k,n-k+1)}
   = \frac{B(k+\sfrac{1}{\alpha}, n-k+1)}{B(k,n-k+1)}.
  \end{equation*}
  Finally, by the transformation $\pi_{(k)}\mapsto a + (b-a)\cdot \pi_{(k)}$, the
  statement is proven.
\end{proof}
\begin{remark}
  An analogous statement holds for the $\Be_{[a, b]}(1,
  \beta)$-distribution. This follows directly from the fact that if
  $X$ follows a $\Be_{[a,b]}(\alpha, \beta)$-distribution, then the
  linearly transformed variable $Y = a+b - X$ follows a
  $\Be_{[a,b]}(\beta, \alpha)$-distribution. 
\end{remark}
\begin{example}
  We want to investigate responses of length $n = 16$ of an
  SRAM-PUF embedded within an error correction scheme such that up to
  $3$ errors can be corrected. For the sake of simplicity, we will
  assume that the cell-wise error probabilities are distributed
  according to a $\Be_{[0,\sfrac{1}{2}]}(\sfrac{1}{9}, 1)$-distribution (such that the
  mean error rate is $\frac{1}{2}\cdot
  \frac{\sfrac{1}{9}}{\sfrac{1}{9} + 1 } = 0.05$). Note that by
  ``ignoring'' bits of the PUF responses, also the error correction
  scheme gets weakened: for every $2$ ignored bits, the correction
  capacity reduces by $1$. Table~\ref{tab:OS1} contains the expected
  values of the respective order statistics (computed along the lines of
  Proposition~\ref{prop:simplifiedOrderStat}).
  \begin{table}[ht]
    \centering
  \scalebox{0.89}{\begin{tabular}[ht]{c||cccccccc}
      $k$ & $1$ & $2$ & $3$ & $4$ & $5$ & $6$ & $7$ & $8$ \\ \hline
      $\E P_{(k)}$ & $2.45\cdot 10^{-7}$ & $2.45\cdot 10^{-6}$ &
      $1.35\cdot 10^{-5}$ & $5.38\cdot 10^{-5}$ & $0.00017$ & $0.00048$ &
      $0.00122$ & $0.00279$ \\ \hline \hline
      $k$ & $9$ & $10$ & $11$ & $12$ & $13$ & $14$ & $15$ & $16$ \\
      \hline
      $\E P_{(k)}$ & $0.00594$ & $0.01189$ & $0.02260$ &
      $0.04110$ &  $0.07193$ & $0.12173$ & $0.2$ &
      $0.32$ 
    \end{tabular}}
    \caption{Order statistics -- expected values (simplified model).}
    \label{tab:OS1}
  \end{table}
  Furthermore, by simulation we are able to estimate the probability
  that a system failure (i.e.\ more errors than the correction scheme
  can handle) occurs when ignoring the $r$ most unstable cells. The
  results of this simulation (with $100000$ simulated PUFs) can be
  found in Table~\ref{tab:OS2}.
  \begin{table}[ht]
    \centering
    \scalebox{0.9}{\begin{tabular}[ht]{c||ccccccc}
      Ignored cells & $0$ & $1$ & $2$ & $3$ & $4$ & $5$ & $6$ \\ \hline
      Correction capacity & $3$ & $3$ & $2$ & $2$ & $1$ & $1$ & $0$ \\
      avg.\ sys.\ failure prob. & $0.00704$ & $0.00183$ &
      $0.00462$ & $0.00134$ & $0.00579$ & $0.00201$ & $0.02203$ \\
      max.\ sys.\ failure prob. & $0.49622$ & $0.35504$ &  $0.49410$ &
      $0.34375$ &  $0.52352$ &  $0.38755$ &  $0.67622$ 
    \end{tabular}}
    \caption{System failure probabilities (simplified model).}
    \label{tab:OS2}
  \end{table}
\end{example}

However, practically, the simplified model is not as precise as the
statistical model developed in the previous section. Therefore,
we investigate a similar example based on this more
sophisticated model from the previous section. In this case, all
parameters will have to be estimated by simulation.

\begin{example}
  We use the posterior-predictive distribution discussed at the
  beginning of this section to obtain flipping probabilities for the
  simulated SRAM-PUFs with $512$ cells each. We have plotted the
  results of this simulation (with $100000$ simulated SRAM devices) in
  Figure~\ref{fig:OS3}. Remarkably, the mean error rate can be
  reduced very quickly from slightly above $0.05$ to about $0.026$ by
  ignoring the $50$ most unstable bits per device (about $10\%$
  information loss). This demonstrates that the exclusion of unstable
  bits is a viable and practically relevant approach to increase the
  stability of an SRAM-PUF.

  \begin{figure}[ht]
    \centering
    \begin{minipage}[ht]{0.49\linewidth}
      \includegraphics[width=1\linewidth]{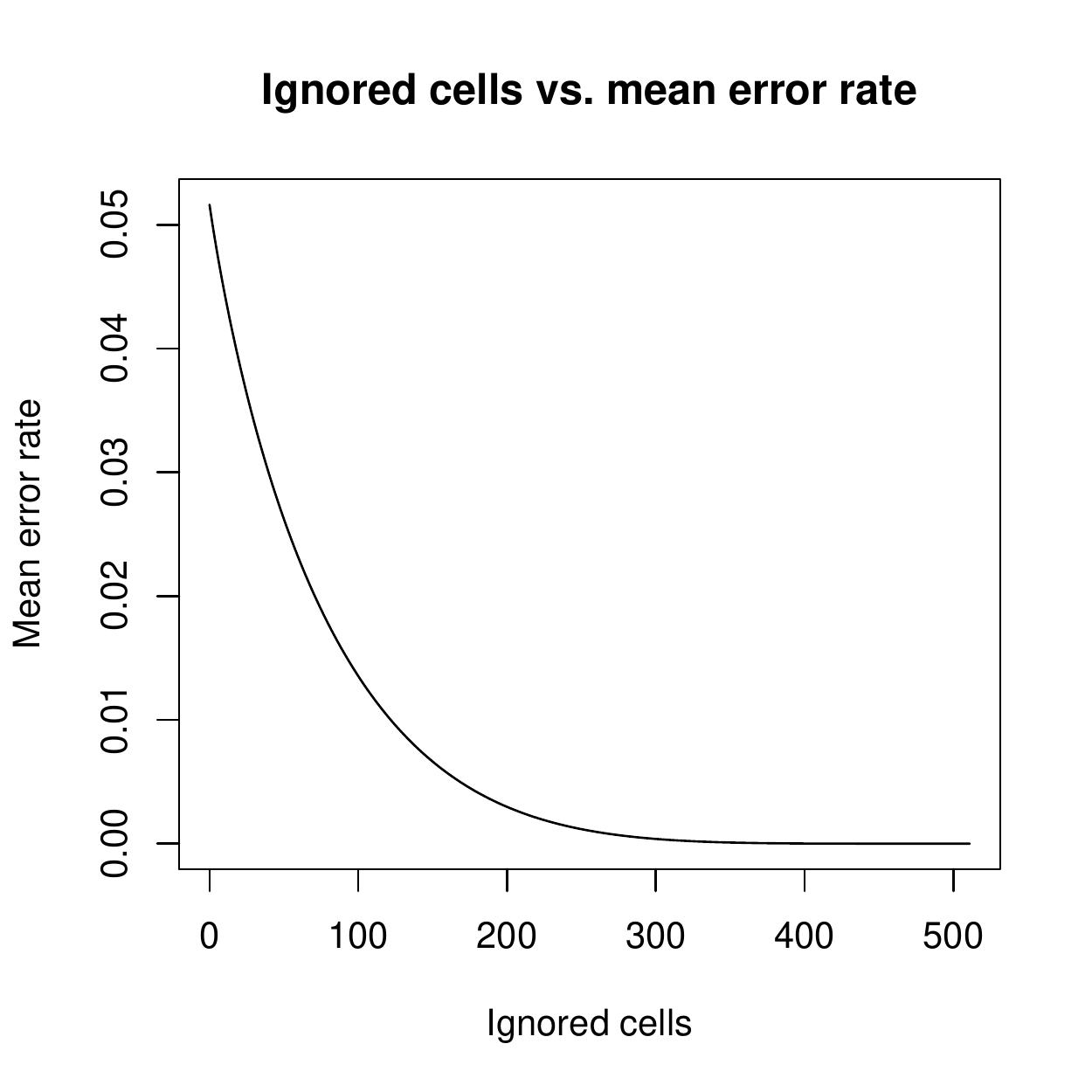}
    \end{minipage}\hfill
    \begin{minipage}[ht]{0.49\linewidth}
      \includegraphics[width=1\linewidth]{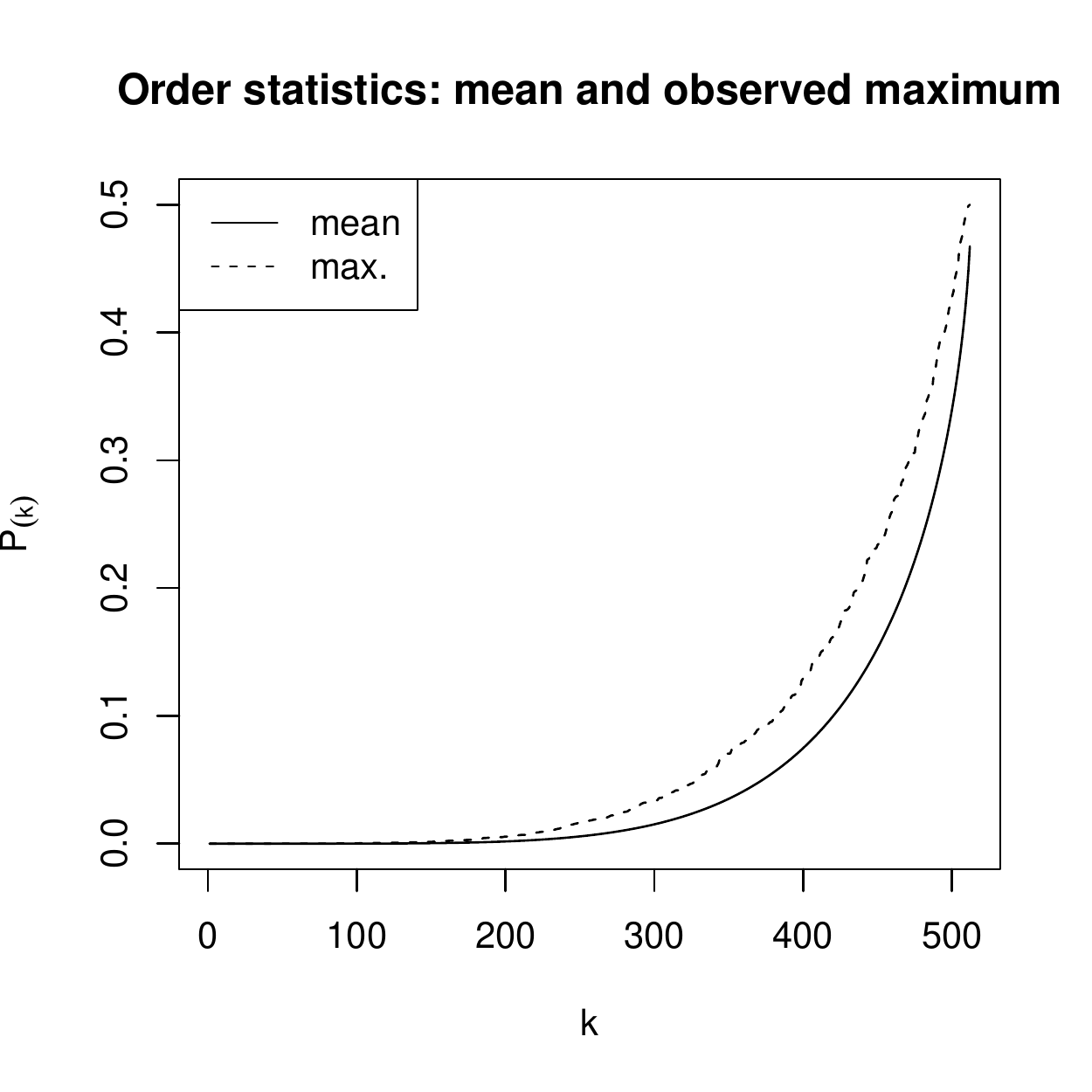}
    \end{minipage}
    \caption{Order statistics and mean error rate.}
    \label{fig:OS3}
  \end{figure}
\end{example}

\section{Conclusion}

The design of error correcting mechanisms capable of correcting the
noise emitted by SRAM-PUFs requires a precise statistical analysis. In
this paper, we presented the framework for such an analysis by proposing
a statistical model which captures the noise behavior of a collection
of SRAM-PUFs (cf.\ Section~\ref{sec:model} and
Section~\ref{sec:appData}). In practice, such a model allows for a
certain flexibility when designing and conducting statistical tests in
the context of quality assurance---which is very important, as these
tests are very expensive in general.

The second tool developed in this paper in order to ascertain precise
predictions for the number of errors in SRAM-PUF responses is the
generalized binomial distribution (cf.\ Section~\ref{sec:gbin}). This
distribution, in combination with the posterior-predictive
distribution for the error probabilities obtained from our given
measurements, permits an evaluation and the design of sufficiently
strong error correction mechanisms (cf.\ Section~\ref{sec:corrAndRed}). Finally, we showed that by ignoring the
most unstable parts of SRAM-PUF responses, the mean error rate could
be reduced significantly. Thus, weaker and simpler error correction
mechanisms could be used.
 
\bibliographystyle{plain}
\bibliography{sramPUFs-stat_paper}

\begin{thebibliography}{10}

\bibitem{orderStatLit}
Mohammad Ahsanullah, Valery~B. Nevzorov, and Mohammad Shakil.
\newblock {\em An Introduction to Order Statistics}, volume~3 of {\em Atlantis
  Studies in Probability and Statistics}.
\newblock Atlantis Press, 2013.

\bibitem{helperData}
Christoph Bösch, Jorge Guajardo, Ahmad-Reza Sadeghi, Jamshid Shokrollahi, and
  Pim Tuyls.
\newblock Efficient {H}elper {D}ata {K}ey {E}xtractor on {FPGAs}.
\newblock In Elisabeth Oswald and Pankaj Rohatgi, editors, {\em Cryptographic
  Hardware and Embedded Systems – CHES 2008}, volume 5154 of {\em Lecture
  Notes in Computer Science}, pages 181--197. Springer Berlin Heidelberg, 2008.

\bibitem{fisz}
Marek Fisz.
\newblock {\em Probability theory and mathematical statistics}.
\newblock Wiley Publications in Statistics. John Wiley \& Sons, Inc., 1965.

\bibitem{BayesianReliability}
Michael~S. Hamada, Alyson~G. Wilson, C.~Shane Reese, and Harry~F. Martz.
\newblock {\em Bayesian Reliability}.
\newblock Springer Series in Statistics. Springer, 2008.

\bibitem{CodingTheory}
Darrel~R. Hankerson, Dean~G. Hoffman, Douglas~A. Leonard, Charles~C. Lindner,
  Kevin~T. Phelps, Chris~A. Rodger, and James~R. Wall.
\newblock {\em Coding Theory And Cryptography---The Essentials}.
\newblock Chapman \& Hall/CRC Pure and Applied Mathematics. CRC Press, 2000.

\bibitem{maxLikPkg}
Arne Henningsen and Ott Toomet.
\newblock max{L}ik: A package for maximum likelihood estimation in {R}.
\newblock {\em Computational Statistics}, 26(3):443--458, 2011.

\bibitem{hoff2009}
Peter~D. Hoff.
\newblock {\em A First Course in Bayesian Statistical Methods}.
\newblock Springer Texts in Statistics. Springer, 2009.

\bibitem{SRAMintroduction}
Daniel~E. Holcomb, Wayne~P. Burleson, and Kevin Fu.
\newblock Initial {SRAM} state as a fingerprint and source of true random
  numbers for {RFID} tags.
\newblock In {\em Proceedings of the Conference on RFID Security}, 2007.

\bibitem{klauer}
Karl~J. Klauer.
\newblock {\em Kriteriumsorientierte Tests: Lehrbuch der Theorie und Praxis
  lehrzielorientierten Messens}.
\newblock Hogrefe, 1987.

\bibitem{genBinomKurz}
Daniel Kurz, Horst Lewitschnig, and J\"urgen Pilz.
\newblock Decision-theoretical model for failures which are tackled by
  countermeasures.
\newblock {\em IEEE Transactions on Reliability}, 63(2):583--592, 2014.

\bibitem{genBinomAppsPkg}
Horst Lewitschnig and David Lenzi.
\newblock {\em GenBinomApps: Clopper-Pearson Confidence Interval and
  Generalized Binomial Distribution}, 2014.
\newblock R package version 1.0-2.

\bibitem{PUFintroduction}
Roel Maes and Ingrid Verbauwhede.
\newblock Physically {U}nclonable {F}unctions: A {S}tudy on the {S}tate of the
  {A}rt and {F}uture {R}esearch {D}irections.
\newblock In Ahmad{-}Reza Sadeghi and David Naccache, editors, {\em Towards
  Hardware-Intrinsic Security: Foundations and Practice}, Information Security
  and Cryptography, pages 3--37. Springer, 2010.

\bibitem{RSoftware}
{R Core Team}.
\newblock {\em R: A Language and Environment for Statistical Computing}.
\newblock R Foundation for Statistical Computing, Vienna, Austria, 2014.

\bibitem{entropyAnalysis}
Robbert van~den Berg, Boris Skoric, and Vincent van~der Leest.
\newblock Bias-based {M}odeling and {E}ntropy {A}nalysis of {PUFs}.
\newblock In {\em Proceedings of the 3rd International Workshop on Trustworthy
  Embedded Devices}, TrustED '13, pages 13--20, New York, NY, USA, 2013. ACM.

\end{thebibliography}

\nocite{*}

\end{document}